\documentclass[sigconf]{acmart}

\usepackage{amsthm}
\usepackage{amsmath,amsfonts}
\usepackage{array}
\usepackage{bbm}
\usepackage{bm}
\usepackage{booktabs}
\usepackage{color}
\usepackage{multicol}
\usepackage{multirow}
\usepackage{xspace}
\usepackage{subfigure}
\usepackage{marvosym}

\newcommand{\eg}{\emph{e.g.,}\xspace}

\newcommand{\ie}{\emph{i.e.,}\xspace}
\newcommand{\ignore}[1]{}
\newcommand{\paratitle}[1]{\vspace{1ex}\noindent\textbf{#1}}

\AtBeginDocument{%
  \providecommand\BibTeX{{%
    \normalfont B\kern-0.5em{\scshape i\kern-0.25em b}\kern-0.8em\TeX}}}

\copyrightyear{2022}
\acmYear{2022}
\setcopyright{acmcopyright}
\acmConference[SIGIR '22] {Proceedings of the 45th International ACM SIGIR Conference on Research and Development in Information Retrieval}{July 11--15, 2022}{Madrid, Spain.}
\acmBooktitle{Proceedings of the 45th International ACM SIGIR Conference on Research and Development in Information Retrieval (SIGIR '22), July 11--15, 2022, Madrid, Spain}
\acmPrice{15.00}
\acmISBN{978-1-4503-8732-3/22/07}
\acmDOI{10.1145/3477495.3531955}

\begin{document}
\fancyhead{}

\title{CORE: Simple and Effective Session-based Recommendation within Consistent Representation Space}

\author{Yupeng Hou$^\dagger$}
\email{houyupeng@ruc.edu.cn}
\affiliation{
    \institution{Gaoling School of Artificial Intelligence, Renmin University of China}
    \country{}
}
\thanks{$\dagger$ Work done during internship at Ant Group.}

\author{Binbin Hu}
\email{bin.hbb@antfin.com}
\affiliation{
    \institution{Ant Group}
    \country{}
}
    
\author{Zhiqiang Zhang}
\email{lingyao.zzq@antfin.com}
\affiliation{
    \institution{Ant Group}
    \country{}
}

\author{Wayne Xin Zhao\textsuperscript{\Letter}}
\email{batmanfly@gmail.com}
\affiliation{
    \institution{Gaoling School of Artificial Intelligence, Renmin University of China}
    \institution{Beijing Key Laboratory of Big Data Management and Analysis Methods}
    \institution{Beijing Academy of Artificial Intelligence}
    \country{}
}
\thanks{\textsuperscript{\Letter} Corresponding author.}

\renewcommand{\authors}{Yupeng Hou, Binbin Hu, Zhiqiang Zhang, Wayne Xin Zhao}
\renewcommand{\shortauthors}{Hou, et al.}

\begin{abstract}
  Session-based Recommendation (SBR) refers to the task of predicting the next item  based on short-term user behaviors within an anonymous session.
  However, session embedding learned by a non-linear encoder is usually not in the same representation space as item embeddings,
  resulting in the inconsistent prediction issue while recommending items.
  To address this issue, we propose a simple and effective framework named \textbf{CORE}, which can unify the representation space for both the encoding and decoding processes.
  Firstly, we design a representation-consistent encoder that takes the linear combination of input item embeddings as session embedding, guaranteeing that sessions and items are in the same representation space.
  Besides, we propose a robust distance measuring method to prevent overfitting of embeddings in the consistent representation space.
  Extensive experiments conducted on five public real-world datasets demonstrate the effectiveness and efficiency of the proposed method.
  The code is available at: \url{https://github.com/RUCAIBox/CORE}.
\end{abstract}



\maketitle

\begin{figure}[t]
    \centering
    \subfigure[]{\label{fig:shift:example}\includegraphics[width=0.225\textwidth]{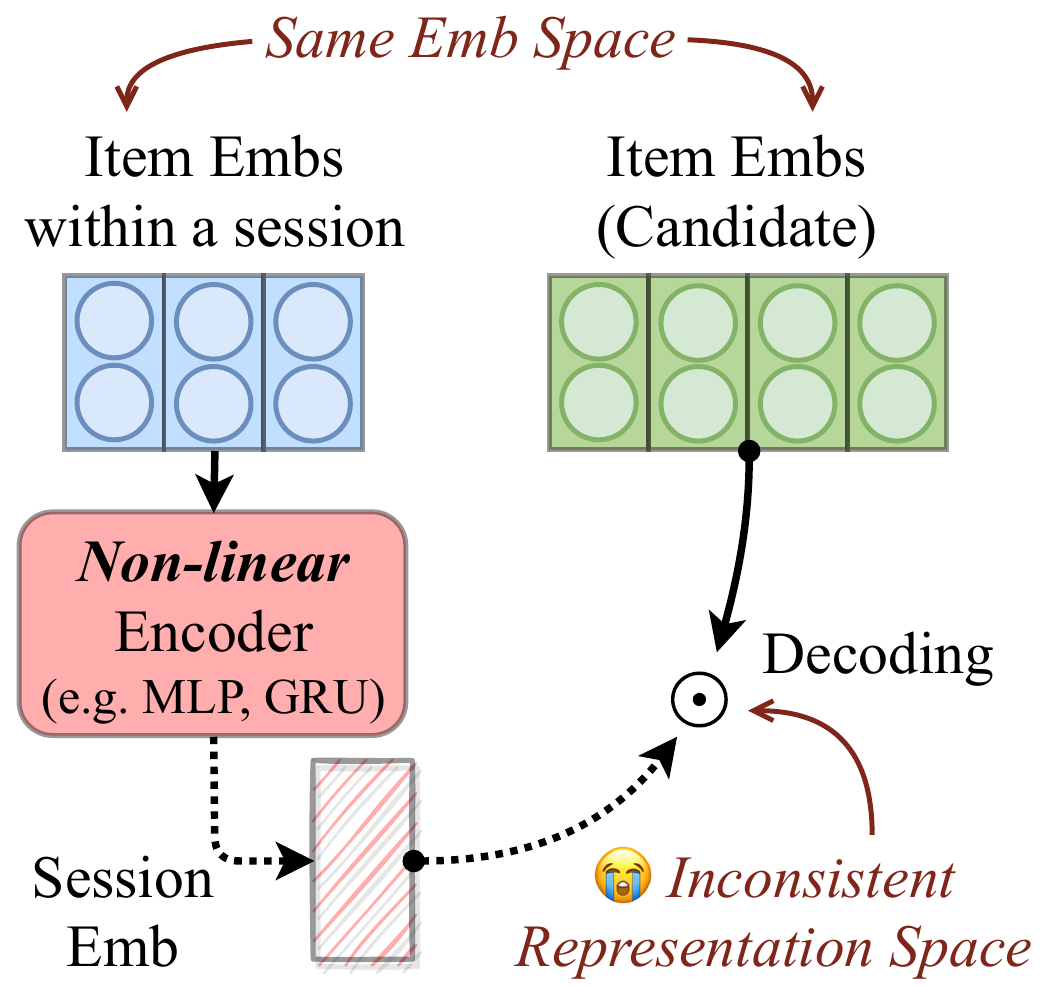}}
    \subfigure[]{\label{fig:shift:tsne}\includegraphics[width=0.215\textwidth]{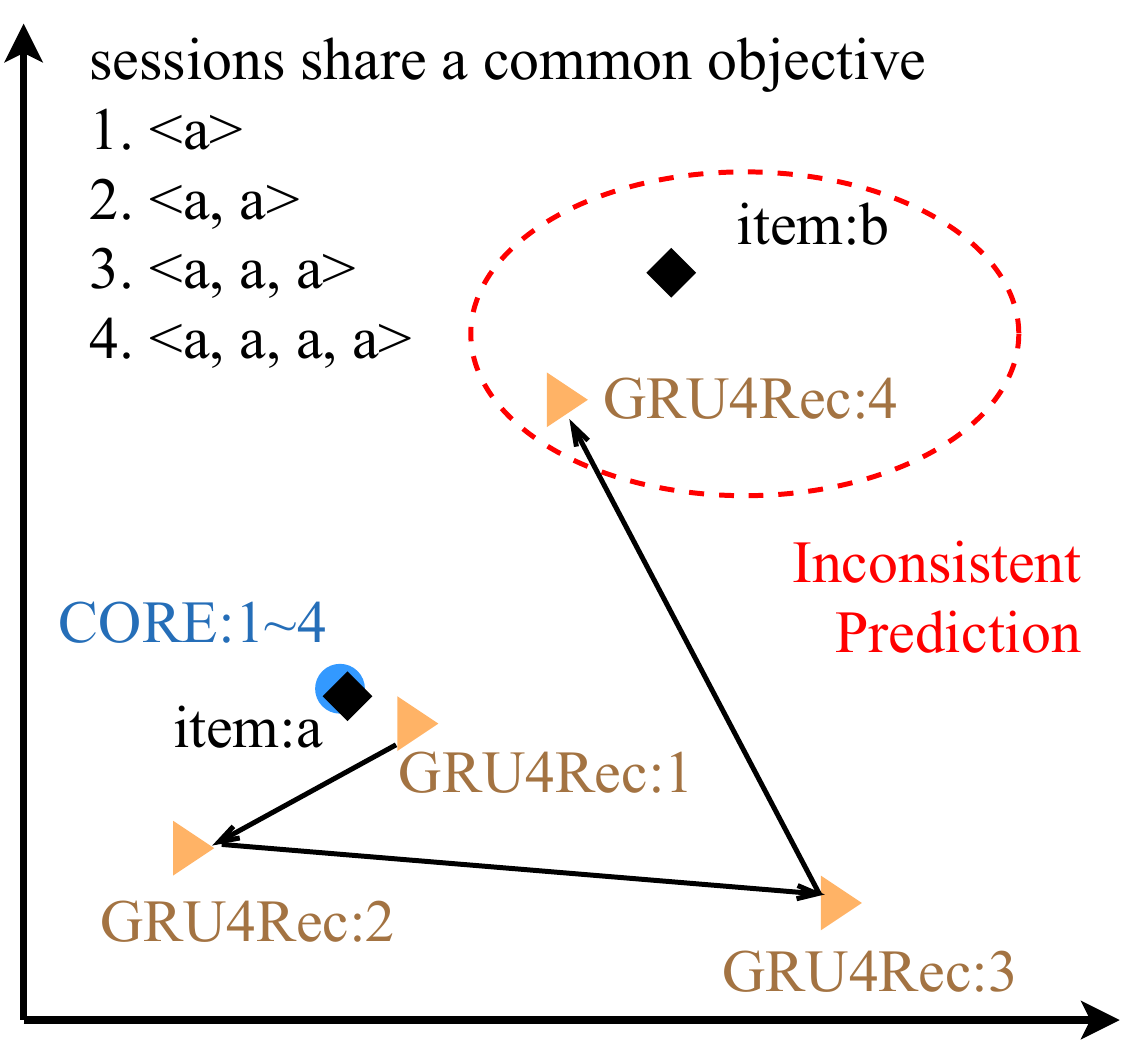}}
    \caption{(a) Encoder-decoder framework of most existing session-based recommendation models and (b) Inconsistent prediction issue while measuring the distance between embeddings for recommending.}
    \label{fig:shift}
\end{figure}

\section{Introduction}

Session-based Recommendation (SBR) aims to capture short-term and dynamic user preferences according to user behaviors (\eg clicked items) within anonymous sessions, which is an important research topic in recommender systems~\cite{SBRsurvey}.
Recent SBR methods usually follow the encoder-decoder framework and
focus on designing effective neural architecture as encoder,
such as Recurrent Neural Network (RNN)~\cite{hidasi2016gru4rec,li2017narm,ren2019repeatnet}, Transformer~\cite{kang2018sasrec,sun2019bert4rec,xie2020cl4rec,zhou20s3rec,bian2021contrastive,he2021locally,zhou2022fmlp} and Graph Neural Network (GNN)~\cite{wu2019srgnn,xu2019gcsan,chen2020lessr,qiu2020fgnn}.
As for the decoder, a widely adopted way is to calculate dot product of session and item embeddings as the interaction probability of the next item.

In line with the focus of the above-mentioned efforts, we notice that 
session embeddings are not usually in the same representation space as item embeddings.
Generally speaking, user behaviors within a short session tend to share a common focus~\cite{chen2020lessr}.
In session-based recommendation, session embeddings are expected to reflect users' short-term preference
and should be similar to embeddings of preferred items.
However, as shown in Figure~\ref{fig:shift:example},
when item embeddings are encoded by non-linear neural networks,
the resulting session embedding doesn't necessarily fall into the space spanned by bases of item embeddings,
and thus it may become unable to  concisely  represent user's preference.
For example,
as illustrated in Figure~\ref{fig:shift:tsne},
for sessions with a common objective (item $a$ clicked multiple times for simulation),
we observe that their embeddings that encoded by non-linear encoder (\eg GRU4Rec~\cite{hidasi2016gru4rec})
fall into different points in embedding space, giving inconsistent prediction when calculating similarities to item embeddings.

Considering the above issues,
we strive to unify the representation space for both encoding and decoding processes in session-based recommendation.
The basic idea is to represent sessions in the item embedding space, for example, by directly summing up item embeddings within the session.
However, in this way, \textcolor{black}{we fail to model the sequential nature, which can be captured by neural encoders like RNNs.}
Thus we consider overcoming two key challenges.
(1) How to design a more suitable encoder, \textcolor{black}{so that we can take advantage of deep non-linear neural networks' immense capability, given that session and item embeddings share a consistent representation space.}
(2) \textcolor{black}{Once the representation space is unified and item embeddings are directly involved in score calculation and model optimization, how to measure the distance between embeddings to avoid overfitting of item embeddings?
}

To this end, we propose a rather simple yet effective framework for session-based recommendation, where session embeddings and item embeddings are in \textbf{CO}nsistent \textbf{RE}presentation space, namely \textbf{CORE}.
Firstly, we encode session embeddings as the linear combination of item embeddings within the session, ensuring the encoded embedding has a consistent representation space as item embeddings.
The weight of each item in the session is learned via a deep non-linear neural network to incorporate various of inductive biases, such as order and importance of items.
Secondly, we revisit
the widely adopt dot product distance measuring in the perspective of optimizing tuplet loss and improve the robustness of distance measuring from multiple aspects.
\textcolor{black}{Extensive experiments on five public datasets demonstrate the effectiveness and efficiency of the proposed framework. 
We also show that performances of existing session-based recommendation models can be significantly improved by the proposed techniques.
}

\section{Methodology}

In this section, we introduce the proposed framework CORE, a simple and effective framework for session-based recommendation within consistent representation space.
Figure~\ref{fig:overall} presents the overall architecture of CORE. 

Now, we start with a brief description of the typical session-based recommendation, which follows an encoder-decoder framework. Firstly, each item is embedded into an unified embedding space.
$\bm{h}_{i} = \operatorname{Emb}(v_i) \in \mathbb{R}^d$ denotes the item embedding for item $v_i$,
where $\operatorname{Emb}(\cdot)$ is the item embedding look-up table and $d$ is the dimension of vectors.
Then we have $\bm{h}_s = \operatorname{Encoder}([\bm{h}_{s,1}, \ldots, \bm{h}_{s,n}]) \in \mathbb{R}^d$ to encode a session $s$ with $n$ items, where $\operatorname{Encoder}(\cdot)$ is usually a non-linear neural network.
Finally, we can predict the  probability distribution for the next item, \ie $\bm{\hat{y}} = \operatorname{Decoder}(\bm{h}_s) \in \mathbb{R}^m$, where $m$ is the number of all items.

\subsection{Representation-Consistent~Encoding}

Here we aim at encoding sessions into item embedding space to overcome the inconsistent representation space issue.
As most existing encoders adopt a non-linear encoder (\eg RNNs or Transformers) directly stacked over the input item embeddings,
the encoded session embeddings are not in the same representation space as items.
In order to keep embeddings within the same space, a natural idea is whether  we can remove non-linear activation functions over item embeddings
and encode sessions as the linear combination of item embeddings.

\begin{figure}[!t]
  \centering
  \includegraphics[width=0.42\textwidth]{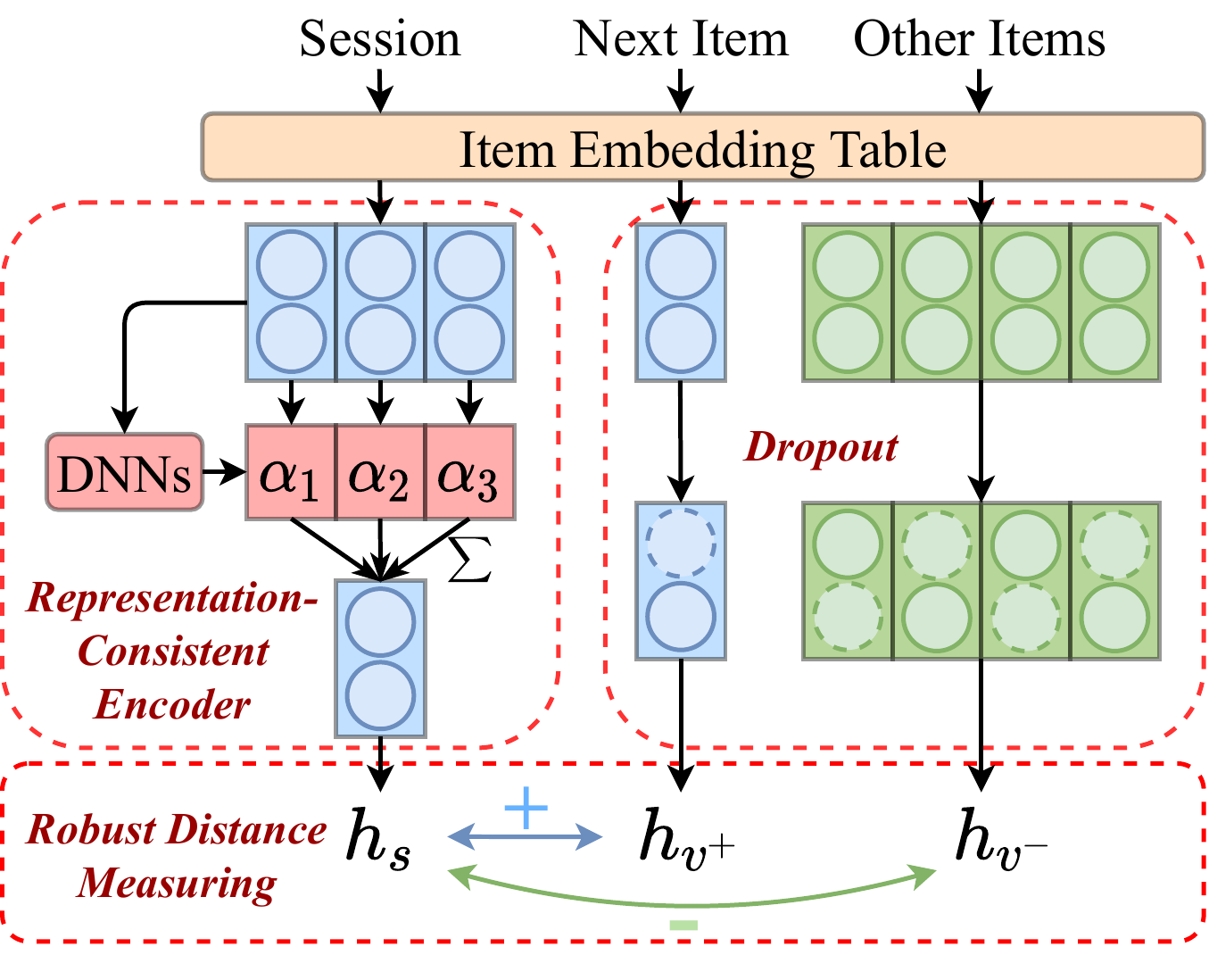}
  \caption{Overall framework of CORE.}
  \label{fig:overall}
\end{figure}

\begin{table*}[!t]
  \caption{Overall performance comparison on five datasets.  ``$*$'' indicates the statistical significance for $p<0.01$ compared to the best baseline method with paired $t$-test. Sessions are split into train/validation/test set in a ratio of 8:1:1 for fair evaluation. We indicate performances of FPMC on Yoochoose as ``$-$'' due to the OOM issue.}
  \label{tab:overall}
  \setlength{\tabcolsep}{1.3mm}{
  \resizebox{2.1\columnwidth}{!}{
  \begin{tabular}{ccccccccccccccr}
  \toprule
    Dataset                       & Metric &  FPMC & GRU4Rec & NARM & SR-GNN & NISER+ & LESSR & SGNN-HN & SASRec & GC-SAN & CL4Rec & CORE-ave & CORE-trm & Improv. \\
    \midrule
    \midrule
    \multirow{2}{*}{Diginetica}   & R@20   & 31.83 &  45.43  & 47.68 & 48.76 & \underline{51.23}  & 48.80 &  50.89  &  49.86 &  50.95 &  50.03 &   50.21  &   \textbf{52.89}*  & +3.24\% \\
                                  & M@20   &  8.79 &  14.77  & 15.58 & 16.93 & \underline{18.32}  & 16.96 &  17.25  &  17.19 &  17.84 &  17.26 &   18.07  &   \textbf{18.58}*  & +1.42\% \\
    \midrule
    \multirow{2}{*}{Nowplaying}   & R@20   & 10.18 &  13.80  & 14.17 & 15.28 & 16.55  & 17.60 &  16.75  &  \underline{20.69} &  18.30 &  20.59 &   20.31  &   \textbf{21.81}*  & +5.41\% \\
                                  & M@20   &  4.51 &   5.83  &  6.11 &  6.10 &  7.14  &  7.13 &   6.13  &   \textbf{8.14} &   \underline{8.13} &   8.21 &    6.62  &    7.35  & $-$ \\
    \midrule
    \multirow{2}{*}{RetailRocket} & R@20   & 46.04 &  55.32  & 58.65 & 58.71 & \underline{60.36}  & 56.22 &  58.82  &  59.81 &  60.18 &  59.69 &   59.18  &   \textbf{61.85}*  & +2.47\% \\
                                  & M@20   & 21.95 &  33.18  & 34.69 & 36.42 & 37.43  & 37.11 &  35.72  &  36.03 &  36.85 &  35.95 &   \underline{37.52}*  &   \textbf{38.76}*  & +3.55\% \\
    \midrule
    \multirow{2}{*}{Tmall}        & R@20   & 20.30 &  23.25  & 31.67 & 33.65 & 35.97  & 32.45 &  39.14  &  35.82 &  35.32 &  35.59 &   \textbf{44.67}*  &   \underline{44.48}*  & +14.13\% \\
                                  & M@20   & 13.07 &  15.78  & 21.83 & 25.27 & 27.06  & 23.96 &  23.46  &  25.10 &  23.48 &  25.07 &   \textbf{31.85}*  &   \underline{31.72}*  & +17.70\% \\
    \midrule
    \multirow{2}{*}{Yoochoose}    & R@20   &  $-$  &  60.78  & 61.67 & 61.84 & 62.99  & 62.89 &  62.49  &  63.55 &  63.24 &  \underline{63.61} &   58.83  &   \textbf{64.61}*  & +1.57\% \\
                                  & M@20   &  $-$  &  27.27  & 27.82 & 28.15 & \underline{28.98}  & 28.59 &  28.24  &  28.63 &  \textbf{29.00} &  28.73 &   25.05  &   28.24  & $-$ \\
  \bottomrule
\end{tabular}
}
}
\end{table*}

Along this line, we propose a Representation-Consistent Encoder (RCE), whose output session embedding is the weighted summarization of item embeddings within a session.
Linear combination guarantees the session embeddings are always in the same embedding space as items.
Although non-linear layers are removed between the encoded session embedding and input item embeddings,
they are essential to incorporating inductive biases and learning weights for input item embeddings.
Formally, we apply arbitrary deep neural networks to learn weights for each item embedding in a session,
\begin{align}
  \bm{\alpha} =& \operatorname{DNN}([\bm{h}_{s,1};\bm{h}_{s,2};\ldots;\bm{h}_{s,n}])\label{eq:alpha},\\
  \bm{h}_s =& \sum_{i=1}^{n} \alpha_i \bm{h}_{s,i},\label{eq:hs}
\end{align}
Then we give two detailed implementations of $\operatorname{DNN}$ without carefully designed architecture.

\subsubsection{Learning Weights via Mean Pooling.}\label{sec:mean_pooling} This variant adopts a mean pooling layer as $\operatorname{DNNs}$, \ie $\alpha_i = \frac{1}{n}$.
This variant ignores the order of items in a session, as well as the importance of each item.

\subsubsection{Learning Weights via Transformer.}\label{sec:transformer} In this variant, we utilize $L$-layers self-attention blocks like SASRec~\cite{kang2018sasrec} as $\operatorname{DNNs}$.
\begin{align}
  \bm{F} = \operatorname{Transformers}([\bm{h}_{s,1};\bm{h}_{s,2};\ldots;\bm{h}_{s,n}]),
\end{align}
where $\bm{F} \in \mathbb{R}^{n\times d'}$ and $d'$ is the output dimension of feed forward network of the last layer of self-attention blocks.
Then we can obtain the normalized weights $\bm{\alpha} \in \mathbb{R}^{n}$,
\begin{align}
  \bm{\alpha} = \operatorname{softmax}(\bm{w} \cdot \bm{F}^\top),
\end{align}
where $\bm{w} \in \mathbb{R}^{d'}$ are learnable parameters. This variant captures the sequential nature via positional encoding technique in Transformer.

\subsection{Robust Distance Measuring for Decoding}

As sessions are encoded as linear combination of item embeddings
and decoded by measuring the distances to items in embedding space,
item embeddings are directly involved in the distance calculation between embeddings, leading to a high risk of overfitting.
Thus we seek a robust way to measure the distance in the unified representation space
to prevent overfitting.
By reviewing the widely adopted dot product distance, we have the following lemma:

\newtheorem{lem}{Lemma}
\begin{lem}
  Given a session embedding $\bm{h}_s$ and item embeddings $\{\bm{h}_{v} | v \in \mathcal{V}\}$,
  when dot product is used to measure the embedding distance,
  optimizing cross entropy loss is approximately proportional to optimize $(N-1)$-tuplet loss~\cite{sohn2016npairloss} with a fixed margin of $2$.
\end{lem}
\begin{proof}
  The lemma can be proved by rewriting the original loss function as,
  \begin{align*}
    \ell_{\text{ori}} =& -\log \frac{\exp(\bm{h}_s\cdot \bm{h}_{v^+})}{\sum_{i=1}^{m}\exp(\bm{h}_s\cdot \bm{h}_{v_i})}\\
    =& \log \left[1 + (|\mathcal{V}|-1)\sum\limits_{v^-\in \mathcal{V}\backslash \{v^+\}}\exp(\bm{h}_s\bm{h}_{v^-} - \bm{h}_s\bm{h}_{v^+})\right],\\
    \simeq& (|\mathcal{V}|-1) \sum\limits_{v^-\in \mathcal{V}\backslash \{v^+\}}\exp(\bm{h}_s\bm{h}_{v^-} - \bm{h}_s\bm{h}_{v^+})\\
    \simeq& (|\mathcal{V}|-1) \sum\limits_{v^-\in \mathcal{V}\backslash \{v^+\}}(\bm{h}_s\bm{h}_{v^-} - \bm{h}_s\bm{h}_{v^+} + 1),\\
    \propto& \sum_{v^-\in \mathcal{V}\backslash \{v^+\}} \left(\|\bm{h}_s - \bm{h}_{v^+}\|^2 - \|\bm{h}_s - \bm{h}_{v^-}\|^2 + 2\right).
  \end{align*}
  where $m$ denotes the number of items, $v^+$ denotes the ground-truth next item for session $s$.
\end{proof}

Given the above lemma, we consider improving the robustness of distance measuring in three aspects. Firstly, data distributions for  different recommendation scenarios may vary a lot, and a fixed margin is inappropriate. Thus, we replace the fixed margin $2$ by a controllable hyper-parameter $\tau$ \textcolor{black}{to suit different scenarios}.
Secondly, we utilize Dropout~\cite{srivastava2014dropout}, a widely adopted technique for robust training, directly over candidate item embeddings. Thirdly,
inspired by recent advances in Contrastive Learning~\cite{chen2020simclr,he2020moco,lin2022ncl},
we propose to measure the distance via cosine distance
for better alignment and uniformity of item embeddings~\cite{wang2020aliANDuni}.
Then we design the loss function with the proposed Robust Distance Measuring (RDM) technique as,
\begin{align}
  \ell = -\log \frac{\exp\left(\cos(\bm{h}_s, \bm{h}'_{v^+})/\tau \right)}{\sum_{i=1}^{m}\exp\left(\cos(\bm{h}_s, \bm{h}'_{v_i})/\tau\right)},
\end{align}
where $\bm{h}'$ denotes the item embeddings with dropout.

\section{Experiments}

\begin{table}
    \caption{Statistics of the datasets.}
    \label{tab:datasets}
    \small
    \resizebox{\columnwidth}{!}{
    \begin{tabular}{lrrrc}
    \toprule
    Dataset & \# Interactions & \# Items & \# Sessions & Avg. Length \\
    \midrule
    \midrule
    Diginetica & 786,582 & 42,862 & 204,532 & 4.12 \\
    Nowplaying & 1,085,410 & 59,593 & 145,612 & 9.21 \\
    RetailRocket & 871,637 & 51,428 & 321,032 & 6.40 \\
    Tmall & 427,797 & 37,367 & 66,909 & 10.62 \\
    Yoochoose & 1,434,349 & 19,690 & 470,477 & 4.64 \\
    \bottomrule
    \end{tabular}
    }%
\end{table}

\paratitle{Datasets and evaluation metrics.}
We conduct experiments on five public datasets collected from real-world platforms:
\emph{Diginetica}, %
\emph{Nowplaying}, %
\emph{RetailRocket}, %
\emph{Tmall} %
and \emph{Yoochoose} %
with their statistics shown in Table~\ref{tab:datasets}.
We filter out sessions of length 1 and items appearing less than 5 times across all datasets~\cite{li2017narm,wu2019srgnn,chen2020lessr},
and split the sessions in each dataset into train/validation/test set in temporal order in a ratio of 8:1:1~\cite{zhao2020revisit}.
To evaluate the performance of different methods, 
we employ two widely-used metrics, top-$20$ Recall (R@$20$) and top-$20$ Mean Reciprocal Rank (M@$20$)~\cite{wu2019srgnn,yu2020tagnn}.

\paratitle{Baselines.}
To evaluate the performance of the proposed method,
we compare it with the following representative baselines:
(1) Matrix factorization based methods:
\underline{\textbf{FPMC}}~\cite{rendle2010fpmc};
(2) RNN-based methods:
\underline{\textbf{GRU4Rec}}~\cite{hidasi2016gru4rec} and
\underline{\textbf{NARM}}~\cite{li2017narm};
(3) GNN-based methods:
\underline{\textbf{SR-GNN}}~\cite{wu2019srgnn},
\underline{\textbf{NISER+}}~\cite{gupta2019niser},
\underline{\textbf{LESSR}}~\cite{chen2020lessr} and
\underline{\textbf{SGNN-HN}}~\cite{pan2020star}.
(4) Transformer-based methods:
\underline{\textbf{SASRec}}~\cite{kang2018sasrec},
\underline{\textbf{GC-SAN}}~\cite{xu2019gcsan} and
\underline{\textbf{CL4Rec}}~\cite{xie2020cl4rec}.
Note that we don't take methods that introduce additional collaborative filtering information~\cite{wang2020gcegnn,jiang2020multiplex,xia2021dhcn,cho2021proxysr} or other side features~\cite{hidasi2016gru4recf} as baselines.
For CORE, we implement two simple variants using mean pooling (Sec.~\ref{sec:mean_pooling}) and Transformers (Sec.~\ref{sec:transformer}) as $\operatorname{DNNs}$ in Eqn.~\eqref{eq:alpha}, namely \emph{CORE-ave} and \emph{CORE-trm}, respectively.

\paratitle{Implementation details}
The proposed models and all the baselines are implemented based on a popular open-source recommendation library \texttt{RecBole}\footnote{\url{https://recbole.io}}~\cite{zhao2020recbole} and its extension \texttt{RecBole-GNN}\footnote{\url{https://github.com/RUCAIBox/RecBole-GNN}} for easy development and reproduction. 
The dimension of the latent vectors is fixed to $100$, and each session is truncated within a maximum length of $50$.
We optimized all the compared methods using Adam optimizer~\cite{kingma2015adam} with a learning rate of $0.001$,
and adopted early-stop training if the M@$20$ performance on the validation set decreased for $5$ consecutive epochs.
We use a batch size of $2048$ for all methods (except TAGNN, for which we use $100$ due to the large memory consumption).
Other hyper-parameters of baselines are carefully tuned following the suggestions from the original papers and we report
each performance under its optimal settings.
For CORE, 
we apply a grid search for controllable margin $\tau$ among $\{0.01, 0.05, 0.07, 0.1, 1\}$ and dropout ratio $\rho$ among $\{0, 0.1, 0.2\}$.
We finally report metrics on the test set with models that gain the highest performance on the validation set.

\begin{figure}[!t]
    \centering
    \includegraphics[width=0.48\textwidth]{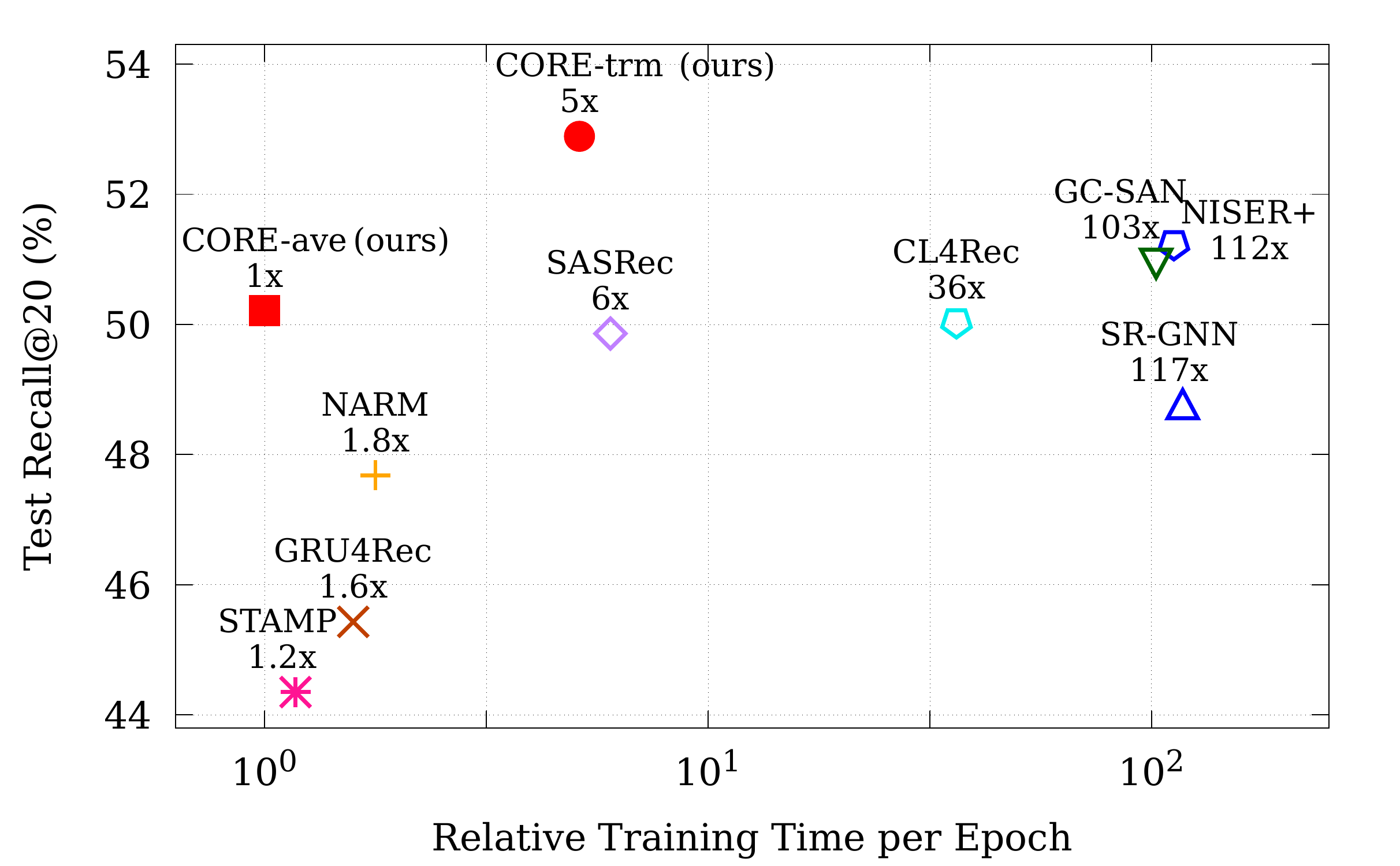}
    \caption{Performances over training time relative to that of CORE-ave on Diginetica.}
    \label{fig:time}
  \end{figure}

\paratitle{Overall comparison.}
From the experimental results in Table~\ref{tab:overall}, we can observe: 
CORE outperforms all the baselines significantly over $8$ of $10$ metrics on the adopted five datasets.
Different from these baselines, CORE
doesn't apply non-linear layers over item embeddings to encode sessions,
but learns weights for each item embedding and adopts a weighted sum for encoding session embeddings and item embeddings in consistent representation space.
Besides, notice that CORE-ave considers neither item order in a session, nor item importances.
However, it still outperforms all the baselines on $3$ metrics, while gaining comparable results with several strong baselines on other metrics.
Without carefully designed encoder architecture, CORE can achieve impressive performance,
which further confirms the importance of encoding session embeddings and item embeddings in consistent representation space.

\paratitle{Analysis 1: Efficiency.}
In Figure~\ref{fig:time}, we plot the performance of several popular session-based recommendation models as well as CORE-trm over their training time per epoch relative to that of CORE-ave on the Diginetica dataset.
We measure the training time on an NVIDIA TITAN V GPU.
As we can see, CORE-ave is the fastest while achieving competitive performance.
By applying mean pooling over sessions,
CORE-ave minimizes memory usage and only
learns a single item embedding table during training.
CORE-trm has a similar training time as SASRec while achieving the best performance among the compared baselines.

\begin{table}[!t]
    \caption{Ablation study of CORE's variants on Diginetica and RetailRocket.}
    \label{tab:ablation}
    \setlength{\tabcolsep}{3.0mm}{
    \begin{tabular}{lcccc}
    \toprule
        \multirow{2} * {Method} & \multicolumn{2}{c}{Diginetica} & \multicolumn{2}{c}{RetailRocket} \\
        & R@20 & M@20 & R@20 & M@20 \\
        \midrule
        \midrule
        CORE & \bm{$52.89$} & \bm{$18.58$}  & \bm{$61.85$}  & \bm{$38.76$} \\
        \ w/o RCE & 49.82 & 17.41 & 59.59 & 36.27 \\
        \ w/o RDM & 52.31 & 18.38 & 60.93 & 37.72 \\
        SASRec & 49.86 & 17.19 & 59.81 & 36.03 \\
    \bottomrule
\end{tabular}
}
\end{table}

\begin{table}[!t]
    \caption{Performance comparison of different methods and their improved variants on two datasets.}
    \label{tab:case_improve}
    \centering
     \setlength{\tabcolsep}{3.0mm}{
    \begin{tabular}{lcccc}
    \toprule
        \multirow{2} * {Method} & \multicolumn{2}{c}{Diginetica} & \multicolumn{2}{c}{RetailRocket} \\
        & R@20 & M@20 & R@20 & M@20 \\
        \midrule
        \midrule
        NARM & 47.68 & 15.58 & 58.65 & 34.69 \\
        \ \ \ + RCE & 51.86 & 18.27 & 60.77 & 37.01 \\
        \ \ \ + RDM & 51.62 & 17.79 & 61.33 & 37.11 \\
        \midrule
        \ \ \ + All & \bm{$52.51$} & \bm{$18.58$} & \bm{$62.19$} & \bm{$38.84$} \\
        \midrule
        SR-GNN & 48.76 & 16.93 & 58.71 & 36.42 \\
        \ \ \ + RCE & 49.51 & 17.53 & 57.05 & 35.70 \\
        \ \ \ + RDM & 51.36 & 18.57 & 61.41 & 38.27 \\
        \midrule
        \ \ \ + All & \bm{$52.38$} & \bm{$18.95$} & \bm{$61.43$} & \bm{$38.38$} \\
    \bottomrule
\end{tabular}
 }
\end{table}

\paratitle{Analysis 2: Ablation study.}
CORE involves several components (\ie RCE and RDM) and we now analyze how each part contributes to the performance.
As SASRec and CORE-trm (CORE for simplify) shares the same Transformer architecture, we select SASRec as the base model to compare.
We mainly consider the following variants:
\textbf{\underline{CORE w/o RCE}} means replacing representation-consistent encoder of CORE to SASRec's encoder;
\textbf{\underline{CORE w/o RDM}} means replacing the proposed robust distance measuring techniques to the traditional dot product distance;
In Table~\ref{tab:ablation}, we can observe that the performance order can be summarized as CORE $>$ CORE w/o RDM $>$ CORE w/o RCE $\simeq$ SASRec.
These results indicate that all the parts are useful to improve the final performance.

\paratitle{Analysis 3: Improving existing methods with RCE \& RDM.}
Here we slightly modify several popular existing session-based models,
showing how the proposed components RCE and RDM improve the performances of existing methods.

Take \emph{SR-GNN~\cite{wu2019srgnn} + RCE} as an example, the original encoder can be formulated as $\bm{h}_s = \bm{W}[\bm{h}_{s,n};\bm{h}_g]$, where $\bm{W} \in \mathbb{R}^{d\times 2d}$ are learnable parameters and $\bm{h}_g$ can be seen as the linear combination of item embeddings within a session.
Then we remove $\bm{W}$ and change the encoder to $\bm{h}_s = (\bm{h}_{s,n} + \bm{h}_g) / 2$ and the modified SR-GNN encoder is a variant of the proposed RCE, where the session embedding is linear combination of item embeddings.

Performance comparison between existing works (NARM~\cite{li2017narm} and SR-GNN~\cite{wu2019srgnn}) and their improved variants are shown in Table~\ref{tab:case_improve}.
We can see that generally the combination of RCE and RDM can gain dramatic performance improvements compared to the original models.
As for only adding one of the proposed techniques, RDM consistently improve the performance,
while RCE has a positive effect in most cases.

\begin{figure}[!t]
    \centering
    \subfigure[GRU4Rec]{\label{fig:tsne:gru4rec}\includegraphics[width=0.15\textwidth]{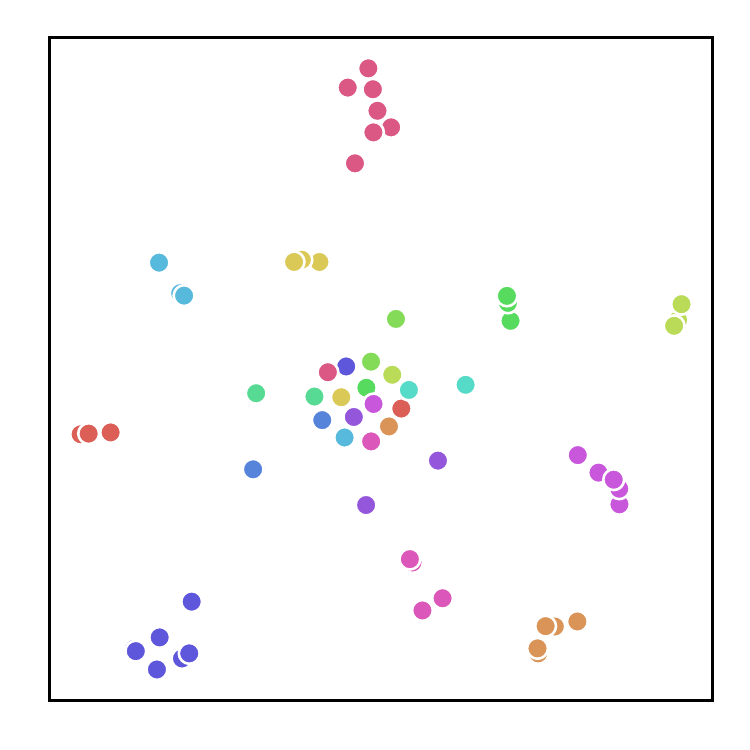}}
    \subfigure[SASRec]{\label{fig:tuning:sasrec}\includegraphics[width=0.15\textwidth]{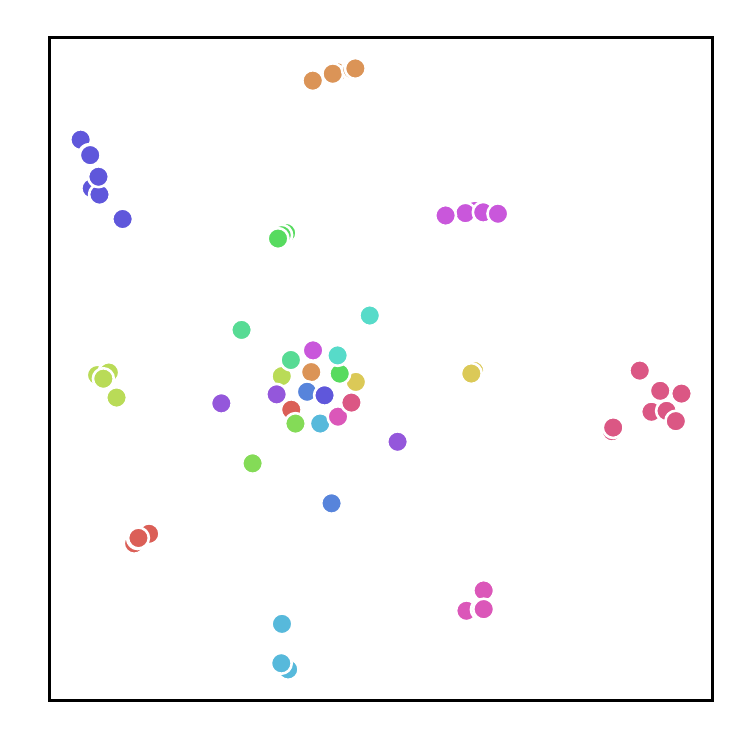}}
    \subfigure[CORE (ours)]{\label{fig:tuning:core}\includegraphics[width=0.15\textwidth]{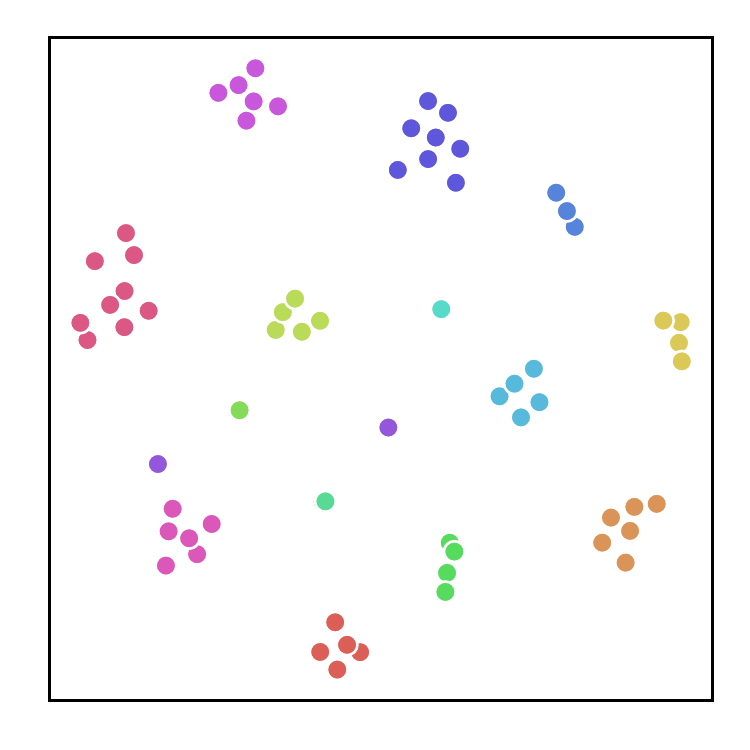}}
    \caption{Visualization of learned session embeddings.}
    \label{fig:tsne}
\end{figure}

\paratitle{Analysis 4: Visualization of session embeddings.}
To show the effectiveness of the proposed model CORE,
we visualize the learned session embeddings using t-SNE~\cite{van2008tsne} algorithm in Figure~\ref{fig:tsne}.
Detailed, sessions with the same next item are viewed as in the same class, and are marked the same color.
We randomly sample $15$ items as ground truth next items.
Then we extract all the corresponding sessions from our test set in Diginetica.
We can see that different classes' session embeddings learned by CORE are well separated from each other compared to those learned by GRU4Rec and SASRec.

\paratitle{Analysis 5: Parameter tuning.} At last, we exime the impact of several importance parameters towards CORE, \ie the margin $\tau$ and the dropout ratio $\rho$. In particular, we vary $\tau$ in $\{0.01, 0.02, \ldots, 0.1\}$ and present the results in Figure~\ref{fig:tuning} (a) and Figure~\ref{fig:tuning} (b). Our method constantly outperforms the best baseline, and achieves the best performance when $\tau = 0.07$ on Diginetica and $\tau = 0.08$ on RetailRocket.
Overall, the performance is stable around $0.04 \le \tau \le 0.08$.
Next,  we vary $\rho$ in the range $0$ and $0.5$ with steps of $0.1$.
As shown in Figure~\ref{fig:tuning} (c) and Figure~\ref{fig:tuning} (d),
when $\rho \le 0.3$, the item dropout increases the robustness of item embedding learning,
while a large ratio of item dropout may hurt the performance, as we can see, the performance decreases sharply when $\rho > 0.3$.

\begin{figure}[!t]
    \centering
    \subfigure[Diginetica]{\label{fig:tuning:dig-tau}\includegraphics[width=0.22\textwidth]{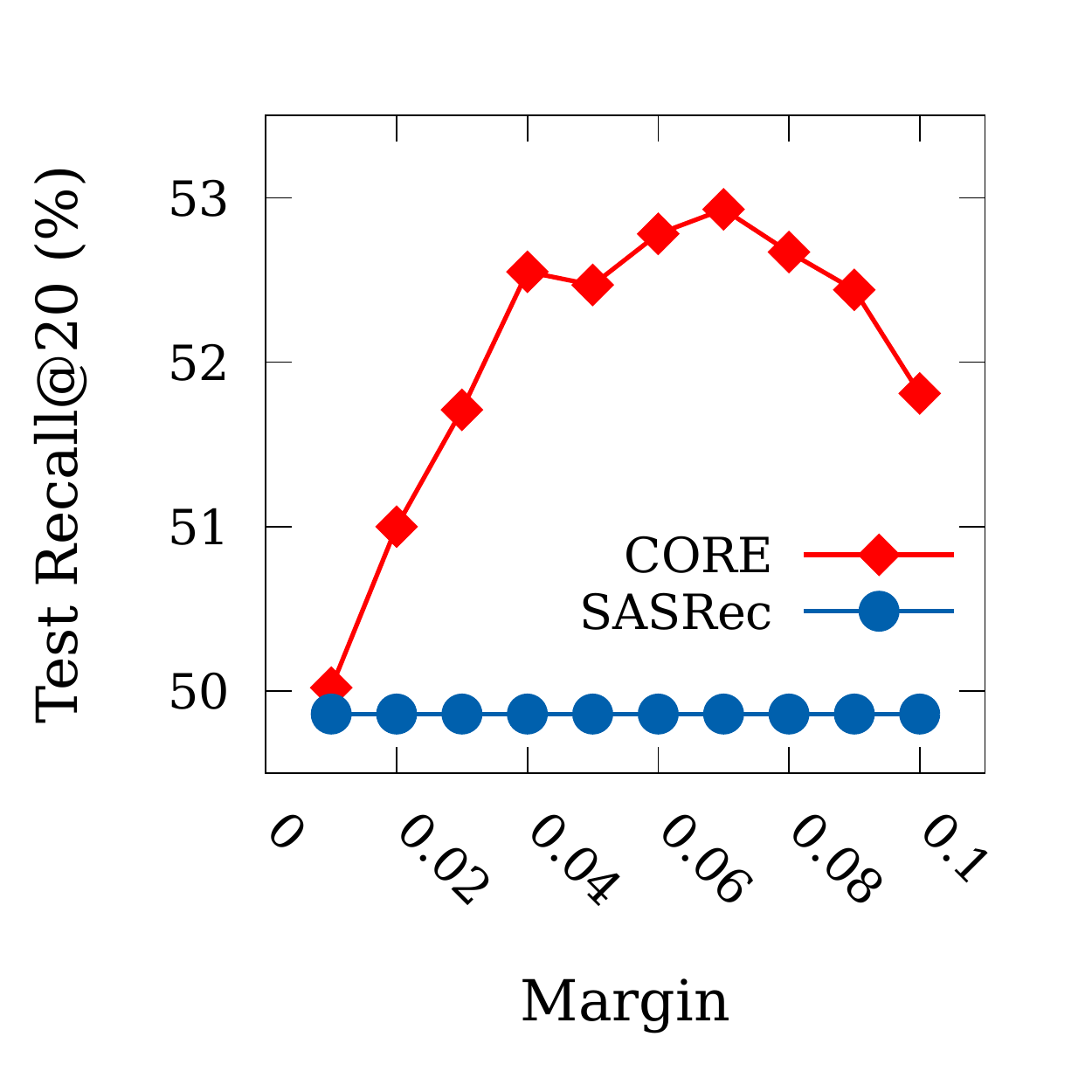}}
    \hspace{0.2cm}
    \subfigure[RetailRocket]{\label{fig:tuning:rr-tau}\includegraphics[width=0.22\textwidth]{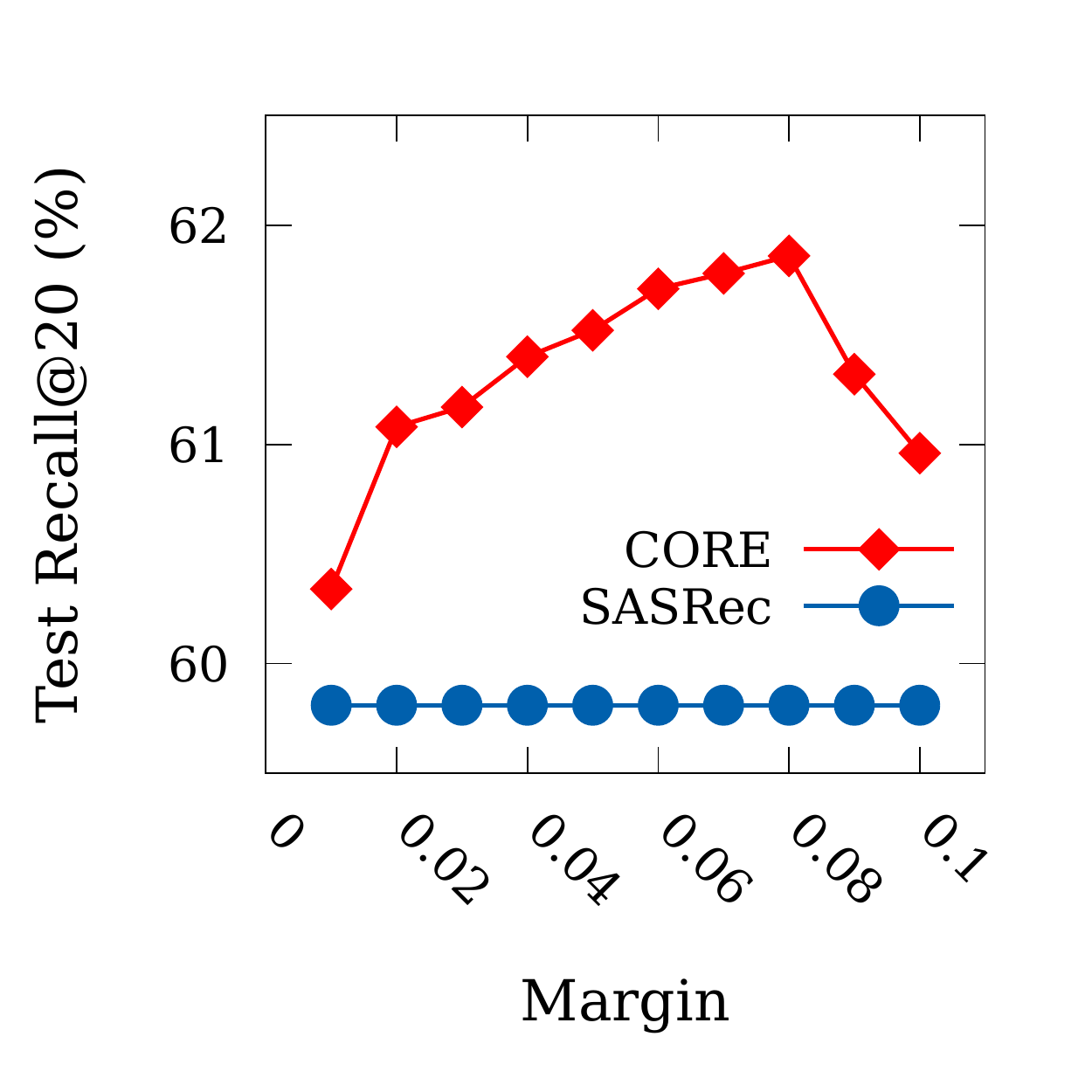}}
    \subfigure[Diginetica]{\label{fig:tuning:dig-drop}\includegraphics[width=0.22\textwidth]{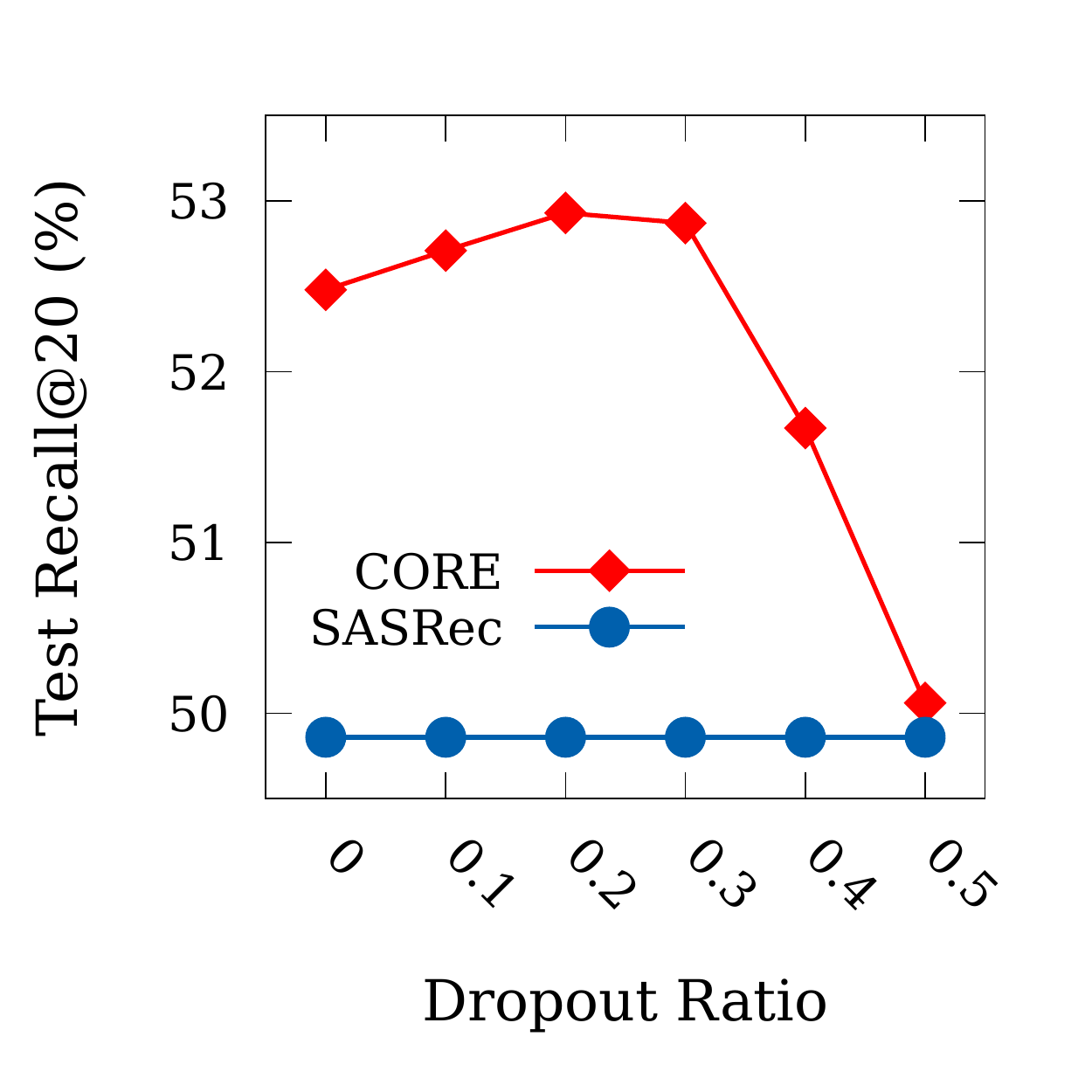}}
    \subfigure[RetailRocket]{\label{fig:tuning:rr-drop}\includegraphics[width=0.22\textwidth]{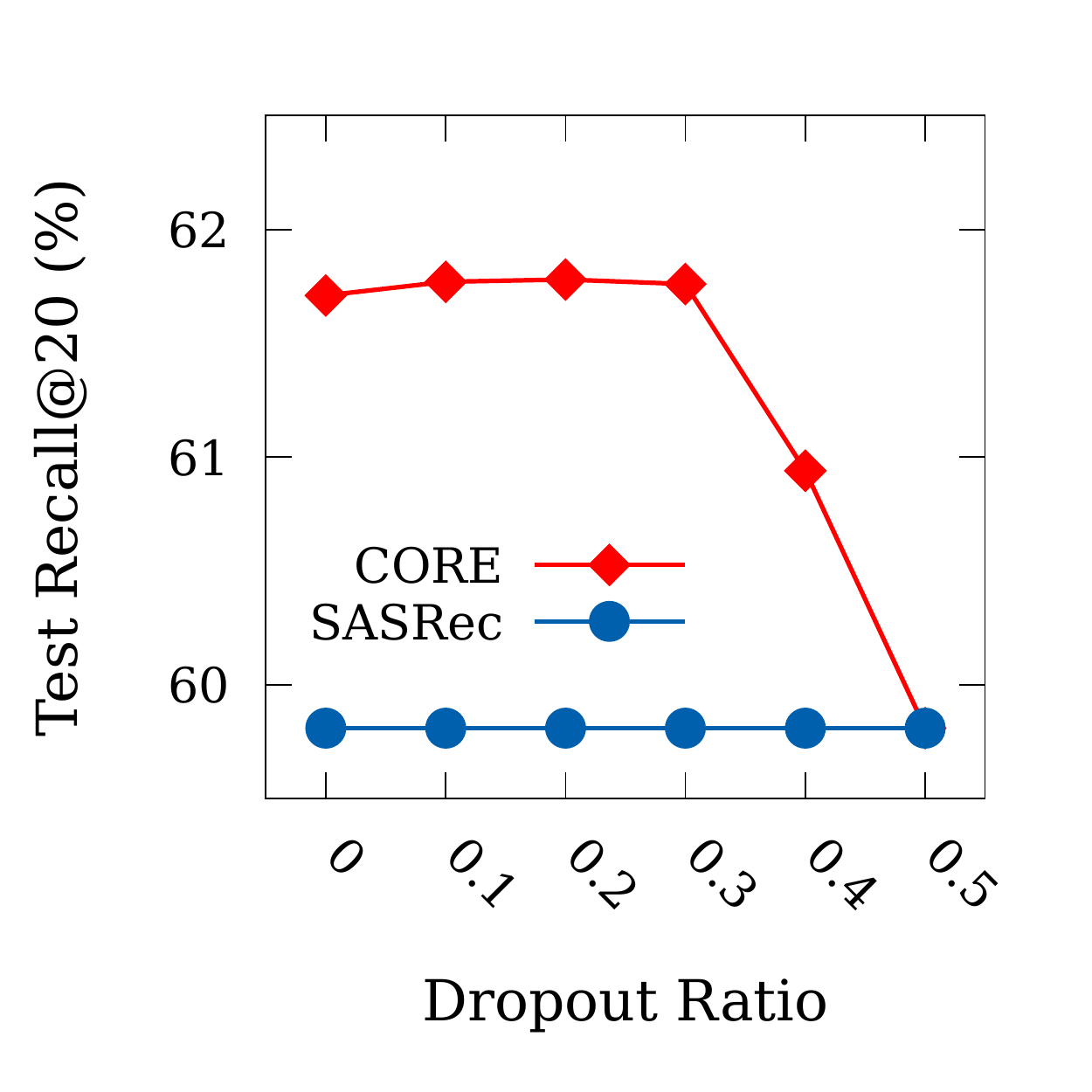}}
    \caption{Parameter tuning of CORE on Diginetica and RetailRocket datasets.}
    \label{fig:tuning}
\end{figure}

\section{Conclusion}
In this paper, we propose CORE, a simple and effective framework for session-based recommendation within consistent representation space, which unifies the representation space throughout encoding and decoding to overcome the inconsistent prediction issue. 
Different from stacking multiple non-linear layers over the item embeddings,
we propose to just apply weighted sum for item embeddings to encode sessions in the consistent representation space as items.
Besdies, we propose robust distance measuring techniques from multiple aspects to prevent overfitting of item embeddings in the proposed framework.
Extensive experiments on five public datasets have shown the effectiveness and efficiency of the proposed approach, as well as how the proposed techniques can help existing methods.

For future work, we will consider studying the expressive ability of the proposed representation-consistent encoder both theoretically and empirically. Besides, we will explore how to introduce side features and useful inductive biases
to the proposed framework.

\begin{acks}
This work was partially supported by the National Natural Science Foundation of China under Grant No. 61872369 and 61832017,
Beijing Outstanding Young Scientist Program under Grant No. BJJWZYJH012019100020098,
and CCF-Ant Group Research Fund.
This work was partially supported by Beijing Academy of Artificial Intelligence (BAAI).
Xin Zhao is the corresponding author.
\end{acks}

\clearpage

\bibliographystyle{ACM-Reference-Format}
\bibliography{main}

\end{document}